\documentclass[a4paper,twoside,10pt]{article}
\usepackage{amssymb}
\usepackage{amsmath}
\usepackage{amsthm}
\usepackage{url}
\usepackage{longtable}
\usepackage[colorlinks,linkcolor=blue,citecolor=green]{hyperref}

\theoremstyle{definition}
\newtheorem{theorem}{Theorem}[section]

\newtheorem{definition}[theorem]{Definition}

\DeclareMathOperator{\tr}{tr}
\DeclareMathOperator{\rk}{rk}
\DeclareMathOperator{\Gr}{Gr}

\newcommand*{\fd}
[2]{\mathchoice{\frac{\delta#1}{\delta#2}}
  {\delta #1/\delta#2}{\delta#1/\delta#2}{\delta#1/\delta#2}}
\newcommand{\ddx}[1]{D^{#1}}

\title{Towards the classification of homogeneous third-order \\
  Hamiltonian operators}

\author{
E.V. Ferapontov$^1$, M.V. Pavlov $^2$, R.F. Vitolo$^3$
\bigskip
\\
\footnotesize $^1$Department of Mathematical Sciences, Loughborough University,
\\
\footnotesize Loughborough, Leicestershire, LE11 3TU, UK
\\
\footnotesize \texttt{e.v.ferapontov@lboro.ac.uk}
\\
\\
\footnotesize $^2$Sector of Mathematical Physics,
\\
\footnotesize Lebedev Physical Institute of Russian Academy of Sciences,
\\
\footnotesize Leninskij Prospekt 53, 119991 Moscow, Russia
\\
\footnotesize and
\\
\footnotesize $^{2}$Department of Applied Mathematics,
\\
\footnotesize National Research Nuclear University MEPHI,
\\
\footnotesize Kashirskoe Shosse 31, 115409 Moscow, Russia
\\
\footnotesize  \texttt{mpavlov@itp.ac.ru}
\\
\\
\footnotesize $^3$Dipartimento di Matematica e Fisica ``E. De Giorgi'',
Universit\`a del Salento
\\
\footnotesize Via per Arnesano, 73100 Lecce, Italy
\\
\footnotesize \texttt{raffaele.vitolo@unisalento.it}
}
\date{\itshape \small Dedicated to Professor Sasha Veselov \\
  \small on the occasion of his 60th birthday}

\begin{document}

\maketitle

\begin{abstract}
  Let $V$ be a vector space of dimension $n+1$.  We demonstrate that
  $n$-component third-order Hamiltonian operators of differential-geometric
  type are parametrised by the algebraic variety of elements of rank $n$ in
  $S^2(\Lambda^2V)$ that lie in the kernel of the natural map
  $S^2(\Lambda^2V)\to \Lambda^4V$. Non-equivalent operators correspond to
  different orbits of the natural action of $SL(n+1)$.  Based on this result,
  we obtain a classification of such operators for $n\leq 4$.
  \\[1mm]
  \noindent MSC 2010: 37K05, 37K10, 37K20, 37K25, 14Nxx.
  \\[1mm]
  \noindent Keywords: Hamiltonian Operator, Jacobi Identity, Projective Group,
  Quadratic Complex, Monge Metric.
\end{abstract}

\newpage

\tableofcontents

\section{Introduction}

In this paper we discuss homogeneous
third-order Hamiltonian operators of differential-geometric type,
\begin{equation}
  J=g^{ij}\ddx{3}+b_{k}^{ij}u_{x}^{k}\ddx{2}
  +(c_{k}^{ij}u_{xx}^{k}+c_{km}^{ij}u_{x}^{k}u_{x}^{m})\ddx{}
  +d_{k}^{ij}u_{xxx}^{k}+d_{km}^{ij}u_{xx}^{k}u_{x}^{m}
  +d_{kmn}^{ij}u_{x}^{k}u_{x}^{m}u_{x}^{n}.
  \label{third}
\end{equation}
Here ${\bf u}=(u^1, \dots, u^n)$ are the dependent variables, and all
coefficients  are functions of ${\bf u}$; $D=\frac{d}{dx}$. 
The operator $J$ is Hamiltonian if and only if the
corresponding Poisson bracket,
\begin{displaymath}
  \{F, G\} = \int \fd{F}{u}J^{}\fd{G}{u} dx,
\end{displaymath}
is skew-symmetric and satisfies the Jacobi identity. Operators of type
(\ref{third}) were introduced by Dubrovin and Novikov in \cite{DN, DN2}, and
thoroughly investigated by Potemin \cite {GP91, GP97}, Doyle \cite{Doyle},
Balandin and Potemin \cite{BP}.  We will only consider the non-degenerate case,
$\det g^{ij}\neq 0$. Under point transformations, ${\bf u}={\bf u}(\tilde {\bf
  u})$, the coefficients of (\ref{third}) transform as differential-geometric
objects.  Thus, $g^{ij}$ transforms as a $(2,0)$-tensor, so that its inverse
$g_{ij}$ defines a pseudo-Riemannian metric, the expressions $ -
\frac{1}{3}g_{js}b_{k}^{si}$, $ - \frac{1}{3}g_{js}c_{k}^{si}$, $ -
g_{js}d_{k}^{si}$ transform as Christoffel symbols of affine connections, etc.
% (see\cite{BP,Doyle,GP91,GP97}).
In particular, the last connection, $\Gamma _{jk}^{i}= - g_{js}d_{k}^{si}$,
must be symmetric and flat \cite{Novikov, Doyle, GP91}.  Therefore, there
exists a distinguished coordinate system (flat coordinates) such that
$\Gamma_{jk}^{i}$ vanish. Flat coordinates are determined up to affine
transformations. We will keep for them the same notation $u^i$, note that $u^i$
are nothing but the densities of Casimirs. In the flat coordinates the last
three terms in (\ref{third}) vanish, leading to a simplified expression
\cite{Doyle},
\begin{equation}
  J=\ddx{}\left(g^{ij}\ddx{}+c_{k}^{ij}u_{x}^{k}\right)\ddx{}.  \label{casimir}
\end{equation}
This operator is Hamiltonian if and only if the metric $g_{ij}$ with lower
indices and the objects $c_{ijk}=g_{iq}g_{jp}c_{k}^{pq}$ satisfy the relations
\cite{GP97}:
\begin{subequations}\label{eq:3}
  \begin{gather}
    g_{mn,k}=-c_{mnk}-c_{nmk},\label{eq:4}\\
    c_{mnk}=-c_{mkn},\label{eq:5}
    \\
    c_{mnk}+c_{nkm}+c_{kmn}=0,\label{eq:6}\\
    c_{mnk,l}=-g^{pq}c_{pml}c_{qnk}.  \label{eq:7}
  \end{gather}
\end{subequations}
It was observed in \cite{FPV} that equations \eqref{eq:3} can be rewritten in
terms of the metric $g$ alone: first of all, system \eqref{eq:3} implies $
c_{nkm}=\frac{1}{3}(g_{nm,k}-g_{nk,m})=\frac{1}{3}g_{n[m,k]}, $ and the
elimination of $c$ results in \begin{equation} g_{mk,n}+g_{kn,m}+g_{mn,k}=0,
  \label{Killing}
\end{equation}
\begin{equation}
  \begin{array}{c}
    g_{m[k,n]l}=-\frac{1}{3}g^{pq} g_{p[l,m]}g_{q[k,n]}.
  \end{array}
  \label{nonlin}
\end{equation}
The second observation of \cite{FPV} is that the generic metric
$g=g_{ij}du^idu^j$ satisfying linear subsystem (\ref{Killing}) is a quadratic
form in $du^i$ and $u^jdu^k-u^kdu^j$, explicitly,
\begin{equation}
  \begin{array}{c}
    g_{ij}du^idu^j=a_{ij}du^idu^j + b_{ijk}du^i(u^jdu^k-u^kdu^j) +
    c_{ijkl}(u^idu^j-u^jdu^i)(u^kdu^l-u^ldu^k),
  \end{array}
  \label{Monge}
\end{equation}
where $a_{ij}, \ b_{ijk}, \ c_{ijkl}$
are arbitrary constants.  Since flat coordinates are defined up to affine
transformations, system (\ref{Killing}), (\ref{nonlin}) is invariant under
transformations of the form
$$
\tilde u^i= l^i({\bf u}), ~~~ \tilde g= g,
$$
where $l^i$ are arbitrary linear forms in the flat coordinates, and $\tilde
g=g$ indicates that $g$ transforms as a metric. What is less obvious is that
system (\ref{Killing})-(\ref{nonlin}) is invariant under the bigger group of
projective transformations,
$$
\tilde u^i= \frac{l^i({\bf u})}{l({\bf u})}, ~~~ \tilde g= \frac{g}{l^4({\bf
    u})},
$$
where $l$ is yet another linear form in the flat coordinates.  It was
demonstrated in \cite{FPV} that projective transformations correspond to
reciprocal transformations of Hamiltonian operator (\ref{casimir}).  Note that
the class of metrics (\ref{Monge}), known in projective geometry as the Monge
metrics of quadratic line complexes, is also invariant under projective
transformations.  Recall that a quadratic line complex is a $(2n-3)$-parameter
family of lines in projective space $\mathbb{P}^n$ specified by a single
quadratic equation in the Pl\"ucker coordinates. Fixing a point $p\in
\mathbb{P}^n$ and taking all lines of the complex that pass through $p$ we
obtain a quadratic cone with vertex at $p$. This field of cones supplies
$\mathbb{P}^n$ with a conformal structure whose general form is given by
(\ref{Monge}). The key invariant of a quadratic line complex is its singular
variety defined by the equation
$$
\det g_{ij}=0.
$$
This is the locus where null cones of $g$ degenerate into a pair of
hyperplanes; it is known to be a hypersurface in $\mathbb{P}^n$ of degree
$2n-2$, see \cite{Dolgachev}, Prop. 10.3.3. For $n=2$ the singular variety is a
conic in $\mathbb{P}^2$, for $n=3$ it is the Kummer quartic in $\mathbb{P}^3$,
for $n=4$ the Segre sextic in $\mathbb{P}^4$, etc.  It turns out that singular
varieties of Monge metrics corresponding to homogeneous third-order Hamiltonian
operators degenerate into double hypersurfaces of degree $n-1$ (equivalently,
$\det g_{ij}$ is a complete square, see Theorem \ref{thm1} of Section
\ref{sec:ag}).  The classification of 2- and 3-component operators can be
summarised as follows.

\medskip

\noindent {\bf Two-component case \cite{FPV}}: {\it Modulo projective
  transformations, every 2-component homogeneous third-order Hamiltonian
  operator can be reduced to constant form.  }

\medskip

\noindent {\bf Three-component case \cite{FPV}}: {\it Modulo (complex)
  projective transformations, the metric of every 3-component homogeneous
  third-order Hamiltonian operator can be reduced to one of the 6 canonical
  forms}:
\begin{gather*}
  g^{(1)}=\begin{pmatrix} (u^{2})^{2}+c & -u^{1}u^{2}-u^{3} & 2u^{2} \\
    -u^{1}u^{2}-u^{3} & (u^{1})^{2}+c(u^{3})^{2} & -cu^{2}u^{3}-u^{1} \\ 2u^{2}
    & -cu^{2}u^{3}-u^{1} & c(u^{2})^{2}+1
  \end{pmatrix}, \\ g^{(2)} = \begin{pmatrix}
    (u^{2})^{2}+1 & -u^{1}u^{2}-u^{3} & 2u^{2} \\
    -u^{1}u^{2}-u^{3} & (u^{1})^{2} & -u^{1} \\
    2u^{2} & -u^{1} & 1
  \end{pmatrix}, \quad g^{(3)} = \begin{pmatrix}
    (u^{2})^{2}+1 &  -u^{1}u^{2}&0 \\
    -u^1u^2 & (u^1)^2 & 0 \\
    0 & 0 & 1%
  \end{pmatrix}, \\ g^{(4)}= \begin{pmatrix} -2u^2 & u^1 & 0
    \\
    u^1 & 0 & 0
    \\
    0 & 0 & 1
  \end{pmatrix}, \quad g^{(5)}=\begin{pmatrix} -2u^2 & u^1 & 1
    \\
    u^1 & 1 &0
    \\
    1 & 0 & 0
  \end{pmatrix}, \quad g^{(6)} =
  \begin{pmatrix}
    1 & 0 & 0\\ 0 & 1 & 0\\ 0 & 0 & 1
  \end{pmatrix}.
\end{gather*}
The corresponding singular varieties, $\det g=0$, are as follows 
\begin{itemize}
\item $g^{(1)}, g^{(2)}$: double quadric;
\item $g^{(3)}, g^{(4)}$: two double planes, one of them at infinity;
\item $g^{(5)}, g^{(6)}$: quadruple plane at infinity.
\end{itemize}
 Direct calculations demonstrate that the metrics $g^{(4)}, g^{(5)},
g^{(6)}$ are flat, while $g^{(1)}, g^{(2)}, g^{(3)}$ are not even
conformally flat: they have non-vanishing Cotton tensor.

The structure of the paper is as follows. In Section~\ref{sec:examples} we
discuss some new examples of Hamiltonian PDEs associated with third-order
Hamiltonian operators. Our examples suggest that every operator (\ref{third})
arises as a Hamiltonian structure of some linearly degenerate
non-diagonalisable system of hydrodynamic type.
In Section~\ref{sec:cf} we introduce the normal form of a quadratic complex in
$\mathbb{P}^n$ that generalises the Clebsch normal form in
$\mathbb{P}^3$. Theorem \ref{thm2} of Section \ref{sec:ag} gives a
parametrisation of $n$-component third-order Hamiltonian operators by the
algebraic variety of elements $Q\in S^2(\Lambda^2V)$ of rank $n$ that belong to
the kernel of the natural map $S^2(\Lambda^2V)\to \Lambda^4V$.  The
classification results are summarised in Section~\ref{sec:class-results}. In
particular, for $n=4$ we obtain 32 non-equivalent multi-parameter canonical
forms.

All computations were performed with the REDUCE computer algebra system
\cite{reduce} and its package CDE \cite{cdiff}.

\section{Examples}\label{sec:examples}

Third-order Hamiltonian operators arise in applications in the context of
Monge-Amp\`ere/WDVV equations of 2D topological field theory \cite{FN, KN1,
  KN2, FPV, PV}. In this Section we demonstrate that 3-component operators
(\ref{casimir}) associated with the metrics $g^{(1)} - g^{(5)}$ can be realised
as Hamiltonian structures of certain linearly degenerate non-diagonalisable
systems of hydrodynamic type. We found the systems using compatibility
conditions between operators and vectors of fluxes of hydrodynamic-type systems
in conservative form~\cite{FPV3}.  We emphasise that even though
the corresponding Hamiltonian densities are nonlocal, the existence of
\emph{local} first-order systems with \emph{local} third-order Hamiltonian
structures is a non-trivial fact.

\medskip
\noindent{\bf Example 1: metric $g^{(1)}$.} The  system
  \begin{equation*}
\begin{array}{c}
 \displaystyle    u^1_t=(\alpha u^2+\beta u^3)_x,\\
    \ \\
\displaystyle     u^2_t=\left( \frac{((u^2)^2-c)(\alpha u^2+\beta u^3)+\gamma (1-c(u^2)^2) +\delta (u^1-cu^2u^3)}{u^1u^2-u^3}\right)_x,\\
    \ \\
\displaystyle      u^3_t=\left( \frac{\alpha u^3((u^2)^2-c)+\beta u^3(u^2u^3-cu^1)+\gamma (u^1-cu^2u^3)+\delta ((u^1)^2-c(u^3)^2)}{u^1u^2-u^3}\right)_x,
     \end{array}
    \end{equation*}
    where $\alpha, \beta, \gamma, \delta$ are arbitrary constants, possesses
    third-order Hamiltonian structure (\ref{casimir}) generated by the metric
    $g^{(1)}$ and the nonlocal Hamiltonian,
 $$
\begin{array}{c} H= \int \Big( 
 \frac{1}{2}\alpha
      (2cxu^1 {\ddx{}}^{-1}  u^{2}  
        +u^3({\ddx{}}^{-1}u^{2})^{2}  
        + c   x^{2}u^3
      )
      +  \beta  u^3(1-c^2) {\ddx{}}^{-1}u^{2}  {\ddx{}}^{-1}u^{3}
        \\
        \ \\
      +  \delta(
      xu^1{\ddx{}}^{-1}u^{1}  
        +cu^3{\ddx{}}^{-1}u^{1}  {\ddx{}}^{-1}u^{2}
        +cu^1{\ddx{}}^{-1} u^{2}  {\ddx{}}^{-1}u^{3}  
        +cxu^{3}  {\ddx{}}^{-1}u^3 )
        \\
        \ \\
      +\frac{1}{2}\gamma
      (cu^1({\ddx{}}^{-1}u^{2})^{2}  
        +x^2u^1
        +2cxu^3{\ddx{}}^{-1}  u^{2} 
\Big) dx.
      \end{array}
 $$
One can show that this system is linearly degenerate, and non-diagonalisable
for generic values of parameters (the diagonalisability conditions are
equivalent to $\alpha \delta - \beta \gamma=0$).

\medskip
\noindent{\bf Example 2: metric $g^{(2)}$.} The  system
  \begin{equation*}
\begin{array}{c}
 \displaystyle    u^1_t=(\alpha u^2+\beta u^3)_x,\\
    \ \\
\displaystyle     u^2_t=\left( \frac{((u^2)^2-1)(\alpha u^2+\beta u^3)-(\gamma +\delta u^1)}{u^1u^2-u^3}\right)_x,\\
    \ \\
\displaystyle      u^3_t=\left( \frac{(u^2u^3-u^1)(\alpha u^2+\beta u^3)-u^1(\gamma +\delta u^1)}{u^1u^2-u^3}\right)_x,
     \end{array}
    \end{equation*}
where $\alpha, \beta, \gamma, \delta$ are arbitrary constants, possesses third-order Hamiltonian structure (\ref{casimir}) generated by the metric $g^{(2)}$
 and the nonlocal Hamiltonian,
 $$
 H= \int \left(\frac{1}{2}\alpha u^3 ({\ddx{}}^{-1}u^2)^{2} 
      +  \beta u^3 {\ddx{}}^{-1}u^2  {\ddx{}}^{-1}u^3 
      -\frac{1}{2} \gamma  x^{2}u^1
      -  \delta x u^1 {\ddx{}}^{-1}u^1\right) dx.
 $$
One can show that this system is linearly degenerate, and non-diagonalisable
for generic values of parameters (the diagonalisability conditions are
equivalent to $\alpha \delta - \beta \gamma=0$).

\medskip
\noindent{\bf Example 3: metric $g^{(3)}$.} The system
$$
u^1_t=(u^2+u^3)_x, \quad u^2_t=\left(\frac{u^2(u^2+u^3)-1}{u^1}\right)_x, \quad
u^3_t=u^1_x,
$$
possesses third-order Hamiltonian structure (\ref{casimir}) generated by the
metric $g^{(3)}$ and the nonlocal Hamiltonian,
    \begin{displaymath}
      H=\int
      \left(-{\ddx{}}^{-1}u^1{\ddx{}}^{-1}u^3+xu^1{\ddx{}}^{-1}u^2\right)
      dx.
    \end{displaymath}
Explicitly, $u_t=J{\delta H}/{\delta u}$ 
where
  $$
  J= \ddx{}\left(
    \begin{array}{ccc}
      \ddx{} & \ddx{} \frac{u^2}{u^1} & 0 \\
        \frac{u^2}{u^1} \ddx{}&\frac{(u^2)^2+1}{2(u^1)^2}\ddx{} +\ddx{} \frac{(u^2)^2+1}{2(u^1)^2}& 0 \\
     0 & 0 & \ddx{}
    \end{array}\right)
  \ddx{}.
    $$
Setting $u^1=f_{xxt}$, $u^2=f_{xtt}-f_{xxx}$, $u^3=f_{xxx}$ we obtain
$f_{xxt}^2 -f_{xxx}f_{xtt}+ f_{xtt}^2-f_{xxt}f_{ttt}-1=0$, which is a
particular case of WDVV equation \cite{Dub}; in the present form, it first
appeared in \cite{Aga}. This  third-order Hamiltonian structure is apparently
new.

  \medskip

\noindent {\bf Remark.} Transformations between third-order PDEs and
systems of hydrodynamic type appearing above, and in examples below, were
first proposed by Mokhov in~\cite{OM95,OM98}.

\medskip
\noindent{\bf Example 4: metric $g^{(4)}$.} The  system 
  \begin{equation*}
    % \label{eq:8}
    u^1_t=u^2_x,\quad u^2_t=\left( \frac{(u^2)^2+u^3}{u^1}\right)_x,\quad u^3_t=u^1_x,
  \end{equation*}
  possesses third-order Hamiltonian structure (\ref{casimir}) generated by the
  metric $g^{(4)}$ and the nonlocal Hamiltonian,
    \begin{displaymath}
      H=\int\left(u^2{\ddx{}}^{-1}u^1{\ddx{}}^{-1}u^2 - {\ddx{}}^{-1}u^1
	{\ddx{}}^{-1}u^3\right)dx.
    \end{displaymath}
    Explicitly, $u_t=J{\delta H}/{\delta u}$ where
  $$
  J= \ddx{}\left(
    \begin{array}{ccc}
      0 & \ddx{} \frac{1}{u^1} & 0 \\
      \frac{1}{u^1} \ddx{}&\frac{u^2}{(u^1)^2}\ddx{} +\ddx{} \frac{u^2}{(u^1)^2}& 0 \\
      0 & 0 & \ddx{}
    \end{array}\right)
  \ddx{}.
    $$
Setting $u^1=f_{xxt}$, $u^2=f_{xtt}$, $u^3=f_{xxx}$ we obtain $f_{xxx}
=f_{ttt}f_{xxt}- f_{xtt}^2$, which is equivalent to the WDVV equation
\cite{Dub} under the interchange of $x$ and $t$.  This third-order  Hamiltonian
representation was constructed in \cite{KN1,KN2}.

\medskip
\noindent{\bf Example 5: metric $g^{(5)}$.} The  system 
  \begin{equation*}
    % \label{eq:8}
    u^1_t=u^2_x,\quad u^2_t=u^3_x,\quad u^3_t=((u^2)^2-u^1u^3)_x,
  \end{equation*}
  possesses third-order Hamiltonian structure (\ref{casimir}) generated by the
  metric $g^{(5)}$ and the nonlocal Hamiltonian,
    \begin{displaymath}
      H=-\int\left(  \frac{1}{2}u^1\left({\ddx{}}^{-1}u^2\right)^2 + {\ddx{}}^{-1}u^2
	{\ddx{}}^{-1}u^3\right)dx.
    \end{displaymath}
Explicitly, $u_t=J{\delta H}/{\delta u}$ 
where
  $$
  J= \ddx{}\left(
    \begin{array}{ccc}
      0 & 0 & \displaystyle\ddx{} \\
      0 & \displaystyle\ddx{} & -\displaystyle\ddx{}u^1 \\
      \displaystyle\ddx{} & -u^1\displaystyle\ddx{} &
      \displaystyle\ddx{}u^2 + u^2\ddx{} + u^1 \ddx{} u^1
    \end{array}\right)
  \ddx{}.
    $$
    Setting $u^1=f_{xxx}$, $u^2=f_{xxt}$, $u^3=f_{xtt}$ we obtain $f_{ttt} =
    f_{xxt}^2 - f_{xxx}f_{xtt}$, which is the simplest case of WDVV equations
    \cite{Dub}. This third-order Hamiltonian representation was found in
    \cite{FN}. Note that although examples 4 and 5 are equivalent under the
    interchange of $x$ and $t$, the action of this elementary transformation on
    Hamiltonian structures is nontrivial, in particular, Hamiltonian operators
    from Examples 4 and 5 are not projectively-equivalent.

    \medskip
    
\noindent{\bf Example 6.}
A natural  4-component generalisation of  Example 5 is the following system,
  \begin{equation*}
    % \label{eq:8}
    u^1_t=u^2_x,\quad u^2_t=u^3_x,\quad u^3_t=u^4_x,
    \quad u^4_t=((u^2)^2-u^1u^3)_x,
  \end{equation*}
which possesses the Hamiltonian formulation $u_t=J{\delta H}/{\delta u}$
  with the  third-order Hamiltonian operator
\begin{equation*} 
J =\ddx{}
\begin{pmatrix}
0 & 0 & 0 & \ddx{} \\ 
0 & 0 & \ddx{} & 0 \\ 
0 & \ddx{} & 0 & -\ddx{}u^1 \\ 
\ddx{} & 0 & -u^1\ddx{} & \ddx{}u^2+u^2\ddx{}%
\end{pmatrix}
\ddx{},
\end{equation*}
and the nonlocal Hamiltonian,
\begin{equation*}  
  H =-\int \left( \frac{1}{2}u^1(D^{-1}u^2)^2 + D^{-1}u^2D^{-1}u^4
    + \frac{1}{2}(D^{-1}u^3)^2\right).
\end{equation*}
Setting $u^1=f_{xxxx}$, $u^2=f_{xxxt}$, $u^3=f_{xxtt}$, $u^4=f_{xttt}$ we obtain a
fourth-order Monge-Amp\`ere equation, 
 $f_{tttt} = f_{xxxt}^2 - f_{xxxx}f_{xxtt}$.
 This example possesses a straightforward $n$-component generalisation,
 \begin{equation*}
u^1_{t}=u^2_{x},\text{ \ }u^2_{t}=u^3_{x}, ~ \dots, ~
u^{n-1}_t=u^n_{x},\text{ \ }u^n_{t}=((u^2)^2-u^{1}u^{3})_{x},
\end{equation*}
with the Hamiltonian structure $u_t=J{\delta H}/{\delta u}$ where
\begin{equation*}
J=\ddx{}\left( 
\begin{array}{ccccc}
 &  &  &  & \ddx{} \\ 
 &  &  &  \ddx{} &  \\ 
 &  & \dots &  & 0 \\ 
&  \ddx{} &  & 0 & -\ddx{}u^1 \\ 
 \ddx{} &  & 0 &  -u^1\ddx{} & \ddx{}u^2+u^2\ddx{}
\end{array}%
\right)\ddx{},
\end{equation*}%
%and the Hamiltonian is%
\begin{equation*}
  H=-\int \left(\frac{1}{2}u^{1}(D^{-1}u^{2})^{2}+\frac{1}{2}
    \underset{m=2}{\overset{m}{\sum }}(D^{-1}u^{m})(D^{-1}u^{n+2-m})\right) dx.
\end{equation*}

\medskip

Examples of this section make it tempting to conjecture that every third-order
Hamiltonian operator (\ref{third}) can be realised as a Hamiltonian structure
of some linearly degenerate non-diagonalisable system of hydrodynamic type.

\section{Canonical form of a quadratic line complex}
\label{sec:cf}

In this section we introduce canonical form of a quadratic complex in
$\mathbb{P}^n$ that can be viewed as a generalisation of the Clebsch normal
form in $\mathbb{P}^3$. This form proves to be convenient for the
characterisation of complexes that correspond to third-order Hamiltonian
operators.

Let us recall the basics of the theory of quadratic complexes. Consider
$n$-dimensional projective space $\mathbb{P}^n$ associated with
$(n+1)$-dimensional vector space $V$. Given two points in $\mathbb{P}^n$ with
homogeneous coordinates $u^i$ and $ v^i$, $i=1, \dots, n+1$, the Pl\"ucker
coordinates $p^{ij}$ of the line through them are defined as $p^{ij}=u^iv^j -
u^jv^i$ (equivalently, one can speak of Pl\"ucker coordinates of the
corresponding 2-dimensional subspace in $V$).  These coordinates satisfy a
system of quadratic relations of the form
$p^{ij}p^{kl}+p^{ki}p^{jl}+p^{jk}p^{il}=0$ that define the Pl\"ucker embedding
of the Grassmannian $\Gr_2(V)$ into $\Lambda^2(V)$.
For $n=3$ one has a single quadratic relation, $ p^{12}p^{34} + p^{31}p^{24} +
p^{23}p^{14}=0$, which defines the Pl\"ucker quadric in $\Lambda^2(V^4)$.

A quadratic line complex is defined by an additional quadratic relation in the
Pl\"ucker coordinates.  This specifies a $(2n-3)$-parameter family of lines in
$\mathbb{P}^n$. Fixing a point $p\in {\mathbb{P}}^n$ and taking all lines of
the complex that pass through $p$ one obtains a quadratic cone with vertex at
$p$. This family of cones supplies $\mathbb{P}^n$ with a conformal structure
(Monge metric), whose explicit form can be obtained as follows. Set
$v^i=u^i+du^i$ (think of $v$ as infinitesimally close to $u$), then the
Pl\"ucker coordinates take the form $p^{ij}=u^idu^j-u^jdu^i$. In the affine
chart $u^{n+1}=1, \ du^{n+1}=0$, part of the Pl\"ucker coordinates simplify to
$p^{(n+1) i}=du^i$, and the equation of the complex reduces to (\ref{Monge}).

Let ${\bf p}=p^{ij}$ be the vector of Pl\"ucker coordinates. Let ${\bf
  p}\Omega^{\alpha}{\bf p}^t=0$ be the Pl\"ucker relations defining
$\Gr_2(V)$, and let ${\bf p}Q{\bf p}^t=0$ be the equation of a quadratic
complex (here $\Omega^{\alpha}$, $Q$ are symmetric matrices); note that $ Q$ is
defined up to transformations of the form $Q\to
Q+c_{\alpha}\Omega^{\alpha}$. Remarkably, there exists a canonical choice of
representative within this class.  For $n=3$ we have a unique Pl\"ucker
relation defined by a $6\times 6$ non-degenerate matrix $\Omega$, and one can
fix $Q$ by the constraint $\tr Q\Omega^{-1}=0$.  This is known as the Clebsch
normal form of a quadratic complex in $\mathbb{P}^3$ \cite{Jess}, p. 109.
Although in higher dimensions the matrices $\Omega^{\alpha}$ are no longer
invertible, there is nevertheless an analogue of Clebsch normal form:

\begin{definition} A quadratic form $Q\in S^2(\Lambda^2V)$ is said to be in
  {\it normal form} if $Q$ belongs to the kernel of the natural map
  $S^2(\Lambda^2V)\to \Lambda^4V$.
\end{definition}

This condition, which  can always be achieved via a transformation $Q\to
Q+c_{\alpha}\Omega^{\alpha}$,  fixes the constants $c_{\alpha}$ uniquely.

\medskip

\noindent {\bf Remark.} Let $\Gr_2(V^*)\subset \Lambda^2(V^*)$ be the
Grassmannian in the dual space, specified by quadratic relations ${\bf
  p^*}{\Omega^{\alpha}}^*{\bf p^*}^t=0$. One can show that the canonical
representative $Q$ defined above can be equivalently fixed by the apolarity
conditions $\tr Q {\Omega^{\alpha}}^*=0$.  Thus, every quadratic complex can be
brought to a canonical form such that the corresponding quadratic form $Q$ is
apolar to the Grassmannian $\Gr_2(V^*)\subset \Lambda^2(V^*)$ . We refer to
\cite{Dolgachev}, Chapter 1 for a general discussion of apolarity in algebraic
geometry.

\section{Complexes corresponding to Hamiltonian operators}
\label{sec:ag}

In this section we give invariant characterisation of quadratic complexes that
correspond to third-order Hamiltonian operators.  Let us first recall the
result of Balandin and Potemin \cite{BP} according to which the general
solution of system (\ref{Killing}) - (\ref{nonlin}) is given by the formula
\begin{equation}
g_{ij}=\phi _{\beta \gamma }\psi _{i}^{\beta }\psi _{j}^{\gamma }, \label{tri}%
\end{equation}
where $\phi _{\beta \gamma }$ is a non-degenerate constant symmetric matrix, and
\begin{equation}
\psi _{k}^{\gamma }=\psi _{km}^{\gamma }u^{m}+\omega _{k}^{\gamma }; \label{lin}%
\end{equation}
here $\psi _{km}^{\gamma }$ and $\omega _{k}^{\gamma }$ are constants such that
$\psi _{km}^{\gamma }=-\psi _{mk}^{\gamma }$, and the matrix $\mathbf{\psi
}=\mathbf{\psi }^{\gamma}_k$ is non-degenerate. Furthermore, Jacobi identities
imply that these constants have to satisfy a set of quadratic relations,
\begin{equation}
\phi _{\beta \gamma }(\psi _{is}^{\beta }\psi _{jk}^{\gamma }+\psi
_{js}^{\beta }\psi _{ki}^{\gamma }+\psi _{ks}^{\beta }\psi _{ij}^{\gamma
})=0, \label{ab}%
\end{equation}%
\begin{equation}
\phi _{\beta \gamma }(\omega _{i}^{\beta }\psi _{jk}^{\gamma
}+\omega _{j}^{\beta }\psi _{ki}^{\gamma }+\omega _{k}^{\beta }\psi
_{ij}^{\gamma })=0. \label{zac}%
\end{equation}
For $n=2$ relations (\ref{ab} - \ref{zac}) are vacuous. For $n=3$ there is only
1 equation (\ref{zac}), while equations (\ref{ab}) are vacuous. For $n=4$ we
have 1 equation (\ref{ab}) and 4 equations (\ref{zac}), 5 relations
altogether. In general, the total number of relations (\ref{ab} - \ref{zac})
equals $C^4_{n+1}$.  Let us first rewrite equations (\ref{tri} - \ref{zac}) in
invariant form.  Formula (\ref{tri}) implies
$$
g=g_{ij}du^idu^j=\phi _{\beta \gamma }(\psi _{i}^{\beta }du^i)(\psi _{j}^{\gamma }du^j),
$$
where, due to the skew-symmetry conditions $\psi _{km}^{\gamma }=-\psi
_{mk}^{\gamma }$, each of the expressions $\psi _{i}^{\beta }du^i$ is a linear
combination of differentials $u^jdu^k-u^kdu^j$ and $du^j$; here $i, j, k=1,
\dots, n$. Let us introduce an auxiliary coordinate $u^{n+1}$ and consider the
$(n+1)\times (n+1)$ skew-symmetric matrix $P$ with the entries
$u^adu^b-u^bdu^a$, where $a, b=1, \dots, n+1$. Then $\psi _{i}^{\beta }du^i$
can be represented as $\tr (A^{\beta}P)$ for some $(n+1)\times (n+1)$
skew-symmetric matrix $A^{\beta}$ (on restriction to the affine chart
$u^{n+1}=1$), so that
\begin{equation}
g=\phi _{\beta \gamma }\tr (A^{\beta}P) \tr (A^{\gamma}P)\vert_{u^{n+1}=1};
\label{A}
\end{equation}
one can use any other affine projection, the resulting operators will be
projectively equivalent. Formula (\ref{A}) involves $n$-dimensional subspace
$A=span \langle A^{\beta} \rangle \subset \Lambda^2V$, and an element
$\phi=\phi_{\beta \gamma} A^{\beta} A^{\gamma} \in S^2 A$. Remarkably,
conditions (\ref{ab} - \ref{zac}) simplify to
\begin{equation}
\phi_{\beta \gamma} A^{\beta}\wedge A^{\gamma}=0,
\label{King}
\end{equation}
that is, $\phi$ must lie in the kernel of the natural map $S^2 A\to
\Lambda^4V$. The main results of this Section are as follows.

\begin{theorem}
  \label{thm1} The singular variety of a quadratic complex corresponding to
  $n$-component third-order Hamiltonian operator (\ref{casimir}) is a double
  hypersurface of degree $n-1$.
\end{theorem}

\begin{proof}[Proof of Theorem \ref{thm1}]
  Formula (\ref{tri}) implies $\det g=\det \phi (\det \psi)^{2}$, the value
  $n-1$ for the degree follows from the fact that the singular variety of a
  quadratic complex has degree $2n-2$ \cite{Dolgachev}, Prop. 10.3.3. This can
  also be seen directly, indeed, $\psi _{k}^{\gamma }=\psi _{km}^{\gamma
  }u^{m}+\omega _{k}^{\gamma }$, and it remains to note that $\det(\psi
  _{km}^{\gamma }u^{m})$ vanishes identically due to the skew-symmetry
  condition $\psi _{km}^{\gamma }=-\psi _{mk}^{\gamma }$ (the matrix $\psi
  _{km}^{\gamma }u^{m}$ has zero eigenvalue corresponding to the eigenvector
  $u^k$). Thus, all terms of degree $n$ cancel identically, leaving an
  expression of degree $n-1$.
\end{proof}

\begin{theorem}
  \label{thm2} Quadratic complexes corresponding to $n$-component third-order
  Hamiltonian operators (\ref{casimir}) are in one-to-one correspondence with
  elements $Q\in S^2(\Lambda^2V)$ of rank $n$ that belong to the kernel of the
  map $S^2(\Lambda^2V)\to \Lambda^4V$.
\end{theorem}

\begin{proof}[Proof of Theorem \ref{thm2}]
  The value $n$ for the rank follows from representation (\ref{A}).
  It remains to note that relations (\ref{ab} - \ref{zac}) are identical to
  (\ref{King}).
\end{proof}

To summarise, a quadratic complex corresponds to an $n$-component third-order
Hamiltonian operator if and only if, in its normal form, the associated
quadratic form has rank $n$.

\section{Classification results}
\label{sec:class-results}

Based on formula (\ref{A}), in this section we address the classification of
homogeneous third-order Hamiltonian operators.
Our strategy will be as follows:

\begin{itemize}

\item Classify $n$-dimensional subspaces $A=span \langle A^1, \dots, A^n\rangle$ in $\Lambda^2V$ modulo natural action of  $SL(n+1)$. 
%(that is, $SL(n+1)$-orbits in $Gr_n(\Lambda^2V)$). 
  Remarkably, this problem has been discussed in the context of metabelian Lie
  algebras \cite{Gauger, Gal}, and a complete classification is known for
  $n\leq 4 $. For $n=1, 2, 3, 4$ the total number of non-equivalent canonical
  forms equals $1, 1, 5, 38$, respectively. Apparently, for $n\geq 5$ the
  problem becomes `wild', and no classification is available.

\item For every subspace $A$ obtained at the previous step, reconstruct
  non-degenerate $\phi=\phi_{\beta \gamma} A^{\beta} A^{\gamma}\in S^2(A)$ that
  belong to the kernel of the natural map $S^2A\to \Lambda^4V$, that is, for
  which formula (\ref{King}) holds; note that this condition is {\it linear} in
  $\phi$.
 
  The constraint for $\phi$ can be equivalently reformulated as
  follows. Consider a generic element of $A$, $A(\xi)=A^{\alpha}\xi_{\alpha}$;
  the condition $\rk A(\xi)=2$ is given by the vanishing of the Pfaffians of all
  $4\times 4$ principal minors of $A(\xi)$, in total, $C_{n+1}^4$ quadratic
  relations of the form $\Omega^{s\beta \gamma}\xi_{\beta}\xi_{\gamma}=0$,
  $s=1, \dots, C_{n+1}^4$. The form $\phi$ must be apolar to every $\Omega^s$:
  $\phi_{\beta \gamma}\Omega^{s\beta \gamma}=0$.
  
  Note that in some cases the subspace $A$ may possess a non-trivial stabiliser
  under the action of $SL(n+1)$: this can be used to simplify the form of
  $\phi$.
  
\item Reconstruct the corresponding Monge metric $g$ by formula (\ref{A}).
\end{itemize}

All results below are formulated modulo (complex) projective transformations.
To save space we only present canonical forms for the corresponding Monge
metrics rather than Hamiltonian operators themselves.

\subsection{1-component case}

Every 1-component third-order Hamiltonian operator can be reduced to $\ddx{3}$,
see \cite{GP91, GP97, Doyle}.  This result goes back to \cite{OM85, AV86,
  OM87}.

\subsection{2-component case}
\label{sec:two-component-case}

Similarly, every 2-component operator can be brought to constant coefficient
form.

\begin{theorem}\label{thm4}\cite{FPV}
  Modulo projective transformations, every 2-component homogeneous third-order
  Hamiltonian operator can be reduced to constant coefficient form.
\end{theorem}

\begin{proof}[Proof of Theorem \ref{thm4}]
  For $n=2$ formula (\ref{A}) involves a 2-dimensional subspace $\langle A^1,
  A^2 \rangle$ in $\Lambda^2(V^3)$. Without any loss of generality one can set
  $A^1=e^1\wedge e^3, \ A^2=e^2\wedge e^3$ (here and in what follows we
  identify $e^i\wedge e^j$ with the corresponding skew-symmetric matrix). In
  the affine chart $u^3=1$ this gives $\tr (A^1P)=2du^1, \ \tr (A^2P)=2du^2$,
  so that the Monge metric $g$ given by (\ref{A}) is constant.
\end{proof}

\subsection{3-component case}
\label{sec:three-component-case}

In this case we have 6 canonical forms: 

\begin{theorem}\label{thm5}\cite{FPV} Modulo  projective
  transformations, the metric of every 3-component homogeneous third-order
  Hamiltonian operator (\ref{casimir}) can be reduced to one of the 6 canonical
  forms:
\begin{gather*}
  g^{(1)}=\begin{pmatrix} (u^{2})^{2}+c & -u^{1}u^{2}-u^{3} & 2u^{2} \\
    -u^{1}u^{2}-u^{3} & (u^{1})^{2}+c(u^{3})^{2} & -cu^{2}u^{3}-u^{1} \\ 2u^{2}
    & -cu^{2}u^{3}-u^{1} & c(u^{2})^{2}+1
  \end{pmatrix}, \\ g^{(2)} = \begin{pmatrix}
    (u^{2})^{2}+1 & -u^{1}u^{2}-u^{3} & 2u^{2} \\
    -u^{1}u^{2}-u^{3} & (u^{1})^{2} & -u^{1} \\
    2u^{2} & -u^{1} & 1
  \end{pmatrix}, \quad g^{(3)} = \begin{pmatrix}
    (u^{2})^{2}+1 &  -u^{1}u^{2}&0 \\
    -u^1u^2 & (u^1)^2 & 0 \\
    0 & 0 & 1%
  \end{pmatrix}, \\ g^{(4)}= \begin{pmatrix} -2u^2 & u^1 & 0
    \\
    u^1 & 0 & 0
    \\
    0 & 0 & 1
  \end{pmatrix}, \quad g^{(5)}=\begin{pmatrix} -2u^2 & u^1 & 1
    \\
    u^1 & 1 &0
    \\
    1 & 0 & 0
  \end{pmatrix}, \quad g^{(6)} =
  \begin{pmatrix}
    1 & 0 & 0\\ 0 & 1 & 0\\ 0 & 0 & 1
  \end{pmatrix}.
\end{gather*}
\end{theorem}

Here we sketch 3 proofs of this classification result. Based on different
ideas, they may be of interest on their own. The first proof is based on the
theory of quadratic complexes and their Segre normal forms. The second proof is
based on the classification of 3-dimensional subspaces in $\Lambda^2V^4$ modulo
natural action of $SL(4)$, that is, on the classification of $SL(4)$-orbits in
$\Gr_3(\Lambda^2V^4)$. Finally, the third proof uses  explicit parametrisation
of quadratic forms of rank 3 that are apolar to the Pl\"ucker quadric.

\begin{proof}[First proof of Theorem \ref{thm5}] 
  Let us begin with the necessary information from the theory of quadratic
  complexes.  The Pl\"ucker embedding of the Grassmannian $\Gr_2(V^4)$ into
  $\Lambda^2(V^4)$, identified with the space of $4\times 4$ skew-symmetric
  matrices $P=p^{ij}$, is the Pl\"ucker quadric, $ p^{12}p^{34} + p^{31}p^{24}
  + p^{14}p^{23}=0$ (the Pfaffian of $P$). Let $\Omega$ be the $6\times 6$
  symmetric matrix corresponding to the Pl\"ucker quadric. A quadratic line
  complex is the intersection of the Pl\"ucker quadric with another homogeneous
  quadratic equation in the Pl\"ucker coordinates, defined by a $6\times 6$
  symmetric matrix $Q$.  The key invariant of a quadratic complex is the Jordan
  normal form of the matrix $C=Q\Omega^{-1}$, known as its Segre
  type. According to Theorem \ref{thm2}, the matrix $C$ of a quadratic complex
  that corresponds to a 3-component third-order Hamiltonian operator, satisfies
  the conditions $\rk C=3, \ \tr C=0$, which impose strong constraints on the
  Segre type. Recall that the Segre symbol carries information about the
  number/sizes of Jordan blocks.  Thus, the symbol $[111111]$ indicates that
  the Jordan form of $C$ is diagonal; the symbol $[222]$ indicates that the
  Jordan form of $C$ consists of three $2\times 2$ Jordan blocks, etc.  We will
  also use `refined' Segre symbols with additional round brackets indicating
  coincidences among the eigenvalues of some of the Jordan blocks, e.g., the
  symbol $[(11)(11)(11)]$ denotes the subcase of $[111111]$ with three pairs of
  coinciding eigenvalues, the symbol $[(111)(111)]$ denotes the subcase with
  two triples of coinciding eigenvalues, etc.  Theorem \ref{thm5} was proved in
  \cite{FPV} by going through the list of 11 Segre types of quadratic complexes
  as listed in \cite{Jess}, and selecting those whose Monge metrics fulfil
  (\ref{nonlin}). A shorter and less computational approach is based on the
  remark that the only Segre types compatible with the constraints $\rk C=3, \
  \tr C=0$ are $[(111)111]$, $[(111)12]$, $[11(112)]$, $[(114)]$, $[(123)]$,
  $[(222)]$. These are exactly the 6 cases of Theorem \ref{thm5}. Recall that
  the singular surface of a generic quadratic complex is Kummer's quartic
  surface. According to Theorem \ref{thm1}, for Monge metrics associated with
  third-order Hamiltonian operators this quartic degenerates into a double
  quadric (which may further split into a pair of planes).
\end{proof}

\begin{proof}[Second proof of Theorem \ref{thm5}]
  This proof is based on the classification of 3-dimensional subspaces in
  $\Lambda^2(V^4)$ \cite{Gauger, Gal}. There are 5 canonical forms, that we list in the format $ A=\langle A^1,\ A^2, \ A^3 \rangle$:
$$
\begin{array}{c}
 \langle e^1 \wedge e^2,\ e^1 \wedge e^3,\ e^2\wedge e^3 \rangle,\\
 \langle e^1 \wedge e^4,\ e^2 \wedge e^4,\ e^3\wedge e^4 \rangle,\\
 \langle e^1 \wedge e^2,\ e^2 \wedge e^4,\ e^3\wedge e^4\rangle, \\
\langle e^1 \wedge e^2,\ e^3\wedge e^4,\ e^1\wedge e^3+e^2\wedge e^4 \rangle, \\
 \langle e^1\wedge e^4+e^2\wedge e^3,\ e^2\wedge e^4,\ e^3\wedge e^4\rangle.\\
\end{array}
$$
Modulo permutations of indices, these are the cases 102, 99, 94, 93, 96 in
Table 2 of \cite{Gal}, respectively. Calculating  $\phi=\phi_{\beta \gamma} A^{\beta} A^{\gamma}\in S^2(A)$  that satisfy condition (\ref{King})  we arrive at the corresponding Monge metrics (\ref{A}); in all cases we use the affine projection $u^4=1$.

\medskip

\noindent {\bf Case 1} gives a degenerate metric, and does not correspond to
a non-trivial Hamiltonian operator.

\medskip

\noindent {\bf Case 2} corresponds to the constant metric $g^{(6)}$ (after
the affine projection $u^4=1$).

\medskip

\noindent {\bf Case 3} gives the  metric
$$
\begin{array}{c}
g=a(p^{12})^2+b(p^{24})^2+c(p^{34})^2+2\alpha p^{12}p^{24}+2\beta p^{24}p^{34}=\\
\ \\
a(u^1du^2-u^2du^1)^2+b(du^2)^2+c(du^3)^2-2\alpha (u^1du^2-u^2du^1)du^2+2\beta du^2du^3.
\end{array}
$$
Here $\det g=(abc-a\beta^2-c\alpha^2)(u^2)^2$, so the singular variety consists
of 2 double planes (one of them at infinity).  The subcase $a=0$ is affinely
equivalent to $g^{(4)}$, the general case $a\ne 0$ is affinely equivalent to
$g^{(3)}$.

\medskip

\noindent {\bf Case 4} gives the  metric
$$
\begin{array}{c}
g=a(p^{12})^2+b(p^{34})^2+c(p^{13}+p^{24})^2+2c p^{12}p^{34}+2(\alpha p^{12}+\beta p^{34})(p^{13}+p^{24})=\\
\ \\
a(u^1du^2-u^2du^1)^2+b(du^3)^2+c(u^1du^3-u^3du^1-du^2)^2\\
-2c (u^1du^2-u^2du^1)du^3+2(\alpha(u^1 du^2-u^2du^1)-\beta du^3)(u^1du^3-u^3du^1-du^2).
\end{array}
$$
We have
$
\det g =(abc+2\alpha \beta c-c^3-\alpha^2b-\beta^2a)
(u^{1}  u^{3}
  +u^{2})
^{2},
$
so that the singular variety is a double quadric.  This leads to the cases $g^{(1)}, g^{(2)}$. Here the case of $g^{(2)}$ is 
distinguished by $27\mu^2+\nu^3=0$ where $\mu=abc+2\alpha \beta c-c^3-\alpha^2b-\beta^2a, \ \nu=2\alpha \beta-ab-3c^2$.

\medskip

\noindent {\bf Case 5} gives the  metric
$$
\begin{array}{c}
g=a(p^{24})^2+b(p^{34})^2+2(\alpha p^{24}+\beta p^{34})(p^{14}+p^{23})+ 2\gamma p^{24}p^{34}=\\
\ \\
a(du^2)^2+b(du^3)^2-2(\alpha du^2+\beta du^3)(u^2du^3-u^3du^2-du^1)+2\gamma du^2du^3.
\end{array}
$$
We have $\det g= 2 \alpha \beta \gamma-\alpha^2b-\beta^2a=const$, so that the
singular variety is a quadruple plane at infinity. This metric is affinely
equivalent to $g^{(5)}$.
\end{proof}

\begin{proof}[Third proof of Theorem \ref{thm5}]

 Introducing 
$$
 {\bf p}=( du^1, \ du^2, \ du^3, \ u^2du^3-u^3du^2, \ u^3du^1-u^1du^3, \ u^1du^2-u^2du^1),
 $$
 one can represent a Monge metric in the form $g= {\bf p}Q {\bf p}^t$, where
 $Q$ is a $6\times 6$ symmetric matrix. According to Theorem \ref{thm2}, we
 have $\rk Q=3, \ \tr Q\Omega^{-1}=0$. Thus, $Q$ and $\Omega$ can be represented
 in the form
$$
Q=\left(
\begin{array}{cc}
A & M\\
M^t & M^tA^{-1}M
\end{array}
\right), ~~~ \Omega=\left(
\begin{array}{cc}
0 & E\\
E & 0
\end{array}
\right),
$$
where $A, M$ and $E$ are $3\times 3$ matrices ($A$ is symmetric, $E$ is the
identity matrix). Note that any symmetric matrix $Q$ of rank 3 can be
represented in this form (one can always assume $A$ to be non-degenerate via a
translation of $u^i$). The condition $\tr Q\Omega^{-1}=0$ reduces to $\tr
M=0$. The classification of normal forms is performed modulo transformations
$Q\to XQX^t$ that preserve $\Omega$: $X\Omega X^t=\Omega$. Setting
$$
X=\left(
\begin{array}{cc}
X_1 & X_2\\
X_3 & X_4
\end{array}
\right),
$$
the condition $X\Omega X^t=\Omega$ reduces to $X_2X_1^t+X_1X_2^t=0,\
X_4X_3^t+X_3X_4^t=0,\ X_2X_3^t+X_1X_4^t=E$. Our goal is to bring $Q$ to normal
form by using special transformations of this type. Taking $X_1=X_4=E$,
$X_2=0$, one obtains that $X_3$ must be skew-symmetric. Applying this
transformation to $Q$ one obtains $A\to A, \ M\to M-AX_3$ (note that this
transformation preserves the condition $\tr M=0$). Thus, $A^{-1}M\to
A^{-1}M-X_3$, which allows one to kill the skew-symmetric part of $A^{-1}M$.
Hence, one can assume that $A^{-1}M=B$ is symmetric, so that $M=AB$, and $Q$
takes the form
$$
Q=\left(
\begin{array}{cc}
A & AB\\
BA & BAB
\end{array}
\right);
$$
recall the condition $\tr M=\tr AB=0$. Applying another transformation,
$X_2=X_3=0, \ X_4=(X_1^{-1})^t$, one obtains $A\to X_1AX_1^t, \ B\to
(X_1^{-1})^tBX_1^{-1}$. Thus, both $A^{-1}$ and $B$ transform in the same way,
and one can apply the theory of normal forms of pairs of quadratic
forms. Modulo complex transformations, there are 3 cases (note that in all of
them $A^{-1}=A$).

\medskip

\noindent {\bf Case 1.} In the generic (diagonal) case one has
$$
A=\left(\begin{array}{ccc}
1 & 0&0\\
0& 1&0\\
0&0&1
\end{array}
\right), ~~~ B=\left(
\begin{array}{ccc}
a & 0&0\\
0& b&0\\
0&0&c
\end{array}
\right),
$$
$\tr AB=0$ gives $a+b+c=0$. The corresponding Monge metric takes the form
$$
g=(du^1+a(u^2du^3-u^3du^2))^2+(du^3+b(u^3du^1-u^1du^3))^2+(du^3+c(u^1du^2-u^2du^1))^2.
$$
Here we have 3 cases: the generic case is equivalent to the metric $g^{(1)}$ of
Theorem \ref{thm5}, the case $b=-a, c=0$ corresponds to $g^{(3)}$, and the case
$a=b=c=0$ corresponds to $g^{(6)}$.

\medskip

\noindent {\bf Case 2.} In the  case of one $2\times 2$ Jordan block one has
$$
A=\left(\begin{array}{ccc}
0 & 1&0\\
1& 0&0\\
0&0&1
\end{array}
\right), ~~~ B=\left(
\begin{array}{ccc}
1 & a&0\\
a& 0&0\\
0&0&b
\end{array}
\right),
$$
$\tr AB=0$ gives $2a+b=0$. The corresponding Monge metric takes the form
$$
g=(du^3-2a(u^1du^2-u^2du^1))^2+
$$
$$
2(du^2+a(u^2du^3-u^3du^2))(du^1+a(u^3du^1-u^1du^3)+u^2du^3-u^3du^2).
$$
For $a\ne 0$ this is the case $g^{(2)}$, $a=0$ corresponds to $g^{(5)}$.

\noindent {\bf Case 3.} In the  case of one $3\times 3$ Jordan block one has
$$
A=\left(\begin{array}{ccc}
0 & 0&1\\
0& 1&0\\
1&0&0
\end{array}
\right), ~~~ B=\left(
\begin{array}{ccc}
0 & 1&a\\
1& a&0\\
a&0&0
\end{array}
\right),
$$
$\tr AB=0$ gives $a=0$. The corresponding Monge metric takes the form
$$
g=2du^3(du^1+u^3du^1-u^1du^3)+(du^2+u^2du^3-u^3du^2)^2.
$$
This is  the case $g^{(4)}$.

\end{proof}

\subsection{4-component case}
\label{sec:four-component-case}

For $n=4$ formula (\ref{A}) involves a 4-dimensional subspace $\langle A^1,
\dots, A^4 \rangle$ in the space of $5\times 5$ skew-symmetric forms,
equivalently, a point in the Grassmannian $\Gr_4(\Lambda^2V^5)$. Modulo natural
action of $SL(5)$, the classification of such subspaces was obtained in
\cite{Gal} in the context of metabelian Lie algebras of signature $(5, 4)$.
Altogether, there are 38 non-equivalent normal forms. Fixing one of the normal
forms, one can reconstruct $\phi$ from condition (\ref{King}). For $n=4$ this
gives 5 conditions for the 10 matrix elements of $\phi$, leaving us with the
freedom of at least 5 arbitrary constants. This freedom can be reduced if the
subspace has a non-trivial stabiliser under the action of $SL(5)$. Note that
the requirement of non-degeneracy of $\phi$ eliminates some of the 38 subcases,
leaving 32 canonical forms.

As an example, let us consider the generic case of the classification
\cite{Gal} that corresponds to the subspace
$$
\langle 
e^1\wedge e^2+e^4\wedge e^5, \ \  e^2\wedge e^5+e^3\wedge e^4, \  \ e^1\wedge e^5+e^2\wedge e^4, \ \  e^1\wedge e^4+e^2\wedge e^3
\rangle. 
$$
This subspace has trivial stabiliser, and generates an open dense orbit of
dimension 24 in $\Gr_4(\Lambda^2V^5)$: note that $\dim SL(5)= \dim
\Gr_4(\Lambda^2V^5)=24$. It gives rise to the 5-parameter Monge metric
\[
\begin{array}{c}
g=\varphi_1(p^{12}+p^{45})^2+2\varphi_2 (p^{12}+p^{45})(p^{15}+p^{24})+\varphi_3(p^{15}+p^{24})^2\\
\ \\
+\varphi_4(p^{25}+p^{34})^2+2(\varphi_1+\varphi_3) (p^{25}+p^{34})(p^{14}+p^{23})+\varphi_5(p^{14}+p^{23})^2\\
\ \\
-2\varphi_4(p^{12}+p^{45})(p^{14}+p^{23})-2\varphi_5(p^{12}+p^{45})(p^{25}+p^{34});
\end{array}
\]
without any loss of generality one can use the affine chart $u^5=1$. All
parameters are essential. The singular variety of this metric is a double
cubic,
\[
u^4(u^{1})^{2}  
  +u^{1}  u^{2}  u^{3}
  -(u^{2})^{3}
  -u^{2}
  -u^{3}  u^{4}
  +(u^{4})^{3}
=0.
\]
Table 1 below contains a complete list of Monge metrics/singular varieties
corresponding to normal forms of 4-dimensional subspaces in $\Lambda^2V^5$; the
first column contains a reference to Table 2 of \cite{Gal}.  The last column
gives dimensions of stabilisers of these subspaces under the action of
$SL(5)$. We always use the affine chart $u^5=1$.

% Using the package longtable
%\footnotesize
\begin{center}
  \begin{longtable}{ | l | l | l |}
    \caption{Monge metrics for $n=4$}
    \\
    \hline\multicolumn{1}{|c|}{Subspace in $\Lambda^2(V^5)$} &
    \multicolumn{1}{c|}{Monge metric $g$ / singular variety $ $ } &
    \multicolumn{1}{c|}{Stab}\\ \hline
\endfirsthead

\multicolumn{3}{c}%
{\tablename\ \thetable{} -- continued from previous page}
\\
\hline\multicolumn{1}{|c|}{Subspace in $\Lambda^2(V^5)$} &
\multicolumn{1}{c|}{Monge metric $g$ / singular variety $$}&
\multicolumn{1}{c|}{Stab} \\
\hline
\endhead

\hline \multicolumn{3}{|c|}{Continued on next page} \\
\hline
\endfoot

\endlastfoot

$\displaystyle
\begin{array}{l}
  e^1\wedge e^2+e^4\wedge e^5
  \\
  e^2\wedge e^5+e^3\wedge e^4
  \\
  e^1\wedge e^5+e^2\wedge e^4
  \\
  e^1\wedge e^4+e^2\wedge e^3
  \ \\
\text{(no.\ 11 in \cite{Gal})}
\end{array}
$

&
$\displaystyle
\begin{array}{c}
\\[0.01mm]
  g=\varphi_1(p^{12}+p^{45})^2+2\varphi_2 (p^{12}+p^{45})(p^{15}+p^{24})
  \\
  +\varphi_3(p^{15}+p^{24})^2 + \varphi_4(p^{25}+p^{34})^2
  +2(\varphi_1 + \varphi_3) (p^{25}+p^{34})(p^{14}+p^{23})
  \\
  +\varphi_5(p^{14}+p^{23})^2
  -2\varphi_4(p^{12}+p^{45})(p^{14}+p^{23})
  -2\varphi_5(p^{12}+p^{45})(p^{25}+p^{34})
  \\[3mm]
 \text{Singular variety is  a  double cubic,}
    \\[3mm]
u^4(u^{1})^{2}  
  +u^{1}  u^{2}  u^{3}
  -(u^{2})^{3}
  -u^{2}
  -u^{3}  u^{4}
  +(u^{4})^{3}
=0
\end{array}
$
&
0
\\

\hline

$\displaystyle
\begin{array}{l}
  e^1\wedge e^4+e^3\wedge e^5
  \\
  e^2\wedge e^5+e^3\wedge e^4
  \\
  e^1\wedge e^5+e^2\wedge e^4
  \\
  e^1\wedge e^3
  \ \\
  \text{(no.\ 15 in \cite{Gal})}
\end{array}
$
&
$\displaystyle
\begin{array}{c}
\\[0.01mm]
g=\varphi_1(p^{31})^2 + \varphi_2(p^{42}+p^{51})^2 + 2\varphi_2(p^{52}+p^{43})(p^{53}+p^{41}) 
\\
 \varphi_3(p^{41}+p^{53})^2 + 2\varphi_3(p^{52}+p^{43})(p^{51}+p^{42}) +
 \varphi_4(p^{43}+p^{52})^2 \\
 + 2\varphi_4(p^{53}+p^{41})(p^{51}+p^{42}) +
 2\varphi_{5}
(p^{53}+p^{41})p^{31}
  \\[3mm]
 \text{Singular variety is a double cubic,}
    \\[3mm]
u^1u^2u^4 
  -(u^{1})^2 +u^{2}  u^{3}
  -(u^{3})^{2}u^{4}
=0
\end{array}
$

& 
1
\\
\hline

$\displaystyle
\begin{array}{l}
e^2\wedge e^4+e^3\wedge e^5
\\
e^2\wedge e^5+e^3\wedge e^4
\\
e^1\wedge e^5+e^2\wedge e^3
  \\
e^1\wedge e^2
  \ \\
\text{(no.\ 21 in \cite{Gal})}
\end{array}
$
&
$\displaystyle
\begin{array}{c}
\\[0.01mm]
  g=\varphi_{1}(p^{21})^2 + 2\varphi_{2}(p^{51} + p^{32})p^{21}
  +\varphi_{3}(p^{51}+p^{32})^2 -2\varphi_{3}(p^{53}+p^{42})p^{21}
 \\
 \varphi_{4}(p^{53}+p^{42})^2
+ \varphi_{4}(p^{52}+p^{43})^2 +  2\varphi_{5}(p^{53}+p^{42})(p^{52}+p^{43}) 
  \\[3mm]
 \text{Singular variety  is a  double  cubic,}
    \\[3mm]
u^1u^2u^4 
  -u^{1}  u^{3}+(u^2)^3
  -u^{2}(u^{3})^2
=0
\end{array}
$

& 
2
\\
\hline

$\displaystyle
\begin{array}{l}
e^2\wedge e^5+e^3\wedge e^4
  \\
  e^1\wedge e^5
  \\
  e^1\wedge e^3+e^2\wedge e^3
  \\
  e^2\wedge e^4
    \ \\
\text{(no.\ 26 in \cite{Gal})}

\end{array}
$
&
$\displaystyle
\begin{array}{c}
\\[0.01mm]
  g=\varphi_{1}(p^{42})^2 + \varphi_{2}(p^{32}+p^{31})^2 + 2\varphi_{3}(p^{52}+p^{51} +p^{43})(p^{32}+p^{31}) \\
   + \varphi_{4}(p^{51})^2 +
  2\varphi_{5}(p^{52}+p^{43})p^{42}
  \\[3mm]
 \text{Singular variety is  a  double cubic,}
    \\[3mm]
u^1u^3u^4-u^1u^2 
  -(u^{2})^2 
=0
\end{array}
$
&
2
\\
\hline

$
\displaystyle
\begin{array}{l}
  e^2\wedge e^5
  \\
  e^1\wedge e^5+e^2\wedge e^4\\
  \hphantom{e^1\wedge e^5}+e^3\wedge e^4
  \\
  e^1\wedge e^4+e^2\wedge e^3
  \\
  e^1\wedge e^3
    \ \\
\text{(no.\ 28 in \cite{Gal})}

\end{array}
$
&
$\displaystyle
\begin{array}{c}
\\[0.01mm]
  g=\varphi_{1}(p^{31})^2 + 2\varphi_{2}(p^{41}+p^{32})p^{31}
  + 2\varphi_{3}(p^{51} + p^{43} +
  p^{42})p^{31}
  + \varphi_{3}(p^{41} + p^{32})^2\\
  + 2\varphi_{4}p^{52}p^{31} + 2\varphi_{4}(p^{51}+p^{43}+p^{42})(p^{41}+p^{32}) +
\varphi_{5}(p^{52})^2
  \\[3mm]
\text{Singular variety is a double cubic,}
    \\[3mm]
(u^1)^2-u^1u^2u^4 +
 u^{2}  (u^{3})^2+u^3(u^2)^2
=0
\end{array}
$

&
3
\\
\hline

$\displaystyle
\begin{array}{l}
  e^2\wedge e^5+e^3\wedge e^4
  \\
  e^1\wedge e^5+e^2\wedge e^4
  \\
  e^2\wedge e^3
  \\
  e^1\wedge e^4
    \ \\
\text{(no.\ 31 in \cite{Gal})}

\end{array}
$
&
$\displaystyle
\begin{array}{c}
\\[0.01mm]
  g=\varphi_{1}(p^{41})^2 + \varphi_{2}(p^{32})^2 + 2\varphi_{3}(p^{51}
  + p^{42})p^{41} \\
  + \varphi_{4}(p^{51}+p^{42})^2 +
2 \varphi_{4}(p^{52}+p^{43})p^{41}

 + 2\varphi_{5}(p^{52}
+ p^{43})p^{32}
  \\[3mm]
 \text{Singular variety  is a  double  cubic,}
    \\[3mm]
u^1u^3u^4 +u^1u^2
  -(u^{2})^2u^{4}
=0
\end{array}
$

&
3
\\
\hline

$\displaystyle
\begin{array}{l}
  e^2\wedge e^5+e^3\wedge e^4
  \\
  e^1\wedge e^5+e^2\wedge e^4
  \\
  e^1\wedge e^4+e^2\wedge e^3
  \\
  e^1\wedge e^3
    \ \\
\text{(no.\ 34 in \cite{Gal})}

\end{array}
$
&
$\displaystyle
\begin{array}{c}
\\[0.01mm]
  g=\varphi_{1}(p^{31})^2 + 2\varphi_{2}(p^{41}+
 p^{32})p^{31}
  + \varphi_{3}(p^{41} + p^{32})^2 +
 2 \varphi_{3}(p^{51}+p^{42})p^{31}
  \\
 +
  2\varphi_{4}(p^{52}+p^{43})p^{31}
  + 2\varphi_{4}(p^{51}+p^{42})(p^{41}+p^{32})
  \\
 +
  \varphi_{5}(p^{51}+p^{42})^2
  + 2\varphi_{5}(p^{52}+p^{43})(p^{41} +
 p^{32})
  \\[3mm]
 \text{Singular variety is  a double cubic,}
    \\[3mm]
u^1u^2u^4 +u^1(u^3)^2
  -(u^{1})^2-u^3(u^2)^2
=0
\end{array}
$

&
4
\\
\hline

$\displaystyle
\begin{array}{l}
  e^4\wedge e^5
  \\
  e^1\wedge e^3+e^2\wedge e^4
  \\
  e^1\wedge e^5+e^2\wedge e^3
  \\
  e^1\wedge e^2
    \ \\
\text{(no.\ 35 in \cite{Gal})}

\end{array}
$
&
$\displaystyle
\begin{array}{c}
\\[0.01mm]
  g=\varphi_{1}(p^{21})^2 + 2\varphi_{2}(p^{51} +
p^{32})p^{21}
  + 2\varphi_{3}(p^{42} + p^{31})p^{21} -
  2\varphi_{4}p^{54}p^{21}
  \\
  + 2\varphi_{4}(p^{51}+p^{32})(p^{42}+p^{31}) + \varphi_{5}(p^{54})^2
  \\[3mm]
 \text{Singular variety is a double cubic,}
    \\[3mm]
(u^1)^2
  -u^4(u^{2})^2
=0
\end{array}
$
&
4
\\
\hline

$\displaystyle
\begin{array}{l}
  e^2\wedge e^3+e^4\wedge e^5
  \\
  e^2\wedge e^5+e^3\wedge e^4
  \\
  e^1\wedge e^4
  \\
  e^1\wedge e^3
    \ \\
\text{(no.\ 36 in \cite{Gal})}

\end{array}
$
&
$\displaystyle
\begin{array}{c}
\\[0.01mm]
  g=\varphi_{1}(p^{31})^2 + 2\varphi_{2}p^{41}p^{31} + \varphi_{3}(p^{41})^2
- \varphi_{4}(p^{54}+p^{32})^2 \\
+ 
\varphi_{4}(p^{52}+p^{43})^2 + 2\varphi_{5}(p^{54}+p^{32})(p^{52}+p^{43})
  \\[3mm]
 \text{Singular variety is a double plane and a double quadric,}
    \\[3mm]
u^1=0, ~~~ u^2u^3-u^4=0
\end{array}
$
&
4
\\
\hline
$\displaystyle
\begin{array}{l}
  e^2\wedge e^5+e^3\wedge e^4
  \\
  e^1\wedge e^5+e^2\wedge e^4
  \\
  e^1\wedge e^4+e^2\wedge e^3
  \\
  e^1\wedge e^2
    \ \\
\text{(no.\ 41 in \cite{Gal})}

\end{array}
$
&
$\displaystyle
\begin{array}{c}
\\[0.01mm]
  g=\varphi_{1}(p^{21})^2 + 2\varphi_{2}(p^{41} + 
p^{32})p^{21}
- 2\varphi_{3}(p^{52}+ p^{43})p^{21} +\varphi_3(p^{41}+p^{32})^2
\\
+2\varphi_{4}(p^{51}+p^{42})p^{21}
+\varphi_5(p^{51}+p^{42})^2
+ 2\varphi_{6}(p^{52}+p^{43})(p^{41}+p^{32})
  \\[3mm]
 \text{Singular variety is a double cubic,}
    \\[3mm]
u^1u^2u^3+(u^1)^2u^4-(u^2)^3=0
\\[3mm]
%\text{Modified by Raffaele}
%\\
%(u^2)^2( - u^1 + u^2)
\end{array}
$
&
5
\\
\hline
$\displaystyle
\begin{array}{l}
  e^3\wedge e^5
  \\
  e^1\wedge e^5+e^2\wedge e^3\\
    \hphantom{e^1\wedge e^5}+e^2\wedge e^4
  \\
  e^1\wedge e^4
  \\
  e^1\wedge e^2
    \ \\
\text{(no.\ 44 in \cite{Gal})}

\end{array}
$
&
$\displaystyle
\begin{array}{c}
\\[0.01mm]
  g=\varphi_{1}(p^{21})^2 + 2\varphi_{2}p^{41}p^{21} + \varphi_{3}(p^{41})^2
  + 2\varphi_{4}(p^{51} + p^{42} +
 p^{32})p^{21}
 + \varphi_{5}(p^{53})^2
  \\[3mm]
 \text{Singular variety is a double plane and a double quadric,}
    \\[3mm]
u^1=0, ~~~ u^1+u^2u^3=0
\end{array}
$
&
5
\\
\hline

$\displaystyle
\begin{array}{l}
  e^1\wedge e^5
  \\
  e^1\wedge e^4+e^2\wedge e^3
  \\
  e^1\wedge e^2+e^2\wedge e^4
  \\
  e^3\wedge e^4
    \ \\
\text{(no.\ 45 in \cite{Gal})}

\end{array}
$
&
$\displaystyle
\begin{array}{c}
\\[0.01mm]
  g=\varphi_{1}(p^{43})^2 + 2\varphi_{2}(p^{42}+p^{21})p^{43}
- \varphi_{2}(p^{41}+p^{32})^2 
 + \varphi_{3}(p^{42}+p^{21})^2\\
 + 2\varphi_{4}(p^{41} + p^{32})p^{43} + 2\varphi_{5}(p^{42}+p^{21})(p^{41}+p^{32})
+ \varphi_{6}(p^{51})^2
  \\[3mm]
\text{Singular variety is a double plane and a double quadric,}
    \\[3mm]
u^1=0, ~~~ u^1u^4+u^2u^3-(u^4)^2=0
\end{array}
$
&
5
\\
\hline

$\displaystyle
\begin{array}{l}
  e^2\wedge e^4+e^3\wedge e^5
  \\
e^2\wedge e^3
  \\
  e^1\wedge e^5
  \\
  e^1\wedge e^4
    \ \\
\text{(no.\ 48 in \cite{Gal})}

\end{array}
$
&
$\displaystyle
\begin{array}{c}
\\[0.01mm]
  g=\varphi_{1}(p^{41})^2 + 2\varphi_{2}p^{51}p^{41} + \varphi_{3}(p^{51})^2
+ \varphi_{4}(p^{32})^2 + 2\varphi_{5}(p^{53} + p^{42})p^{32}
  \\[3mm]
\text{Singular variety is a double plane and a double quadric,}
    \\[3mm]
u^1=0, ~~~ u^2u^4+u^3=0
\end{array}
$
&
5
\\
\hline

$\displaystyle
\begin{array}{l}
  e^2\wedge e^5+e^3\wedge e^4
  \\
e^1\wedge e^5+e^2\wedge e^3
  \\
  e^1\wedge e^4
  \\
  e^1\wedge e^3
    \ \\
\text{(no.\ 50 in \cite{Gal})}

\end{array}
$
&
$\displaystyle
\begin{array}{c}
\\[0.01mm]
  g=\varphi_{1}(p^{31})^2 + 2\varphi_{2}p^{41}p^{31} + \varphi_{3}(p^{41})^2
+ 2\varphi_{4}(p^{51}+p^{32})p^{31} \\
+ 2\varphi_{5}(p^{52}+p^{43})p^{31}
 + \varphi_{5}(p^{51}+p^{32})^2
  \\[3mm]
\text{Singular variety is a double plane and a double quadric,}
    \\[3mm]
u^1=0, ~~~ u^1-u^2u^3=0
\end{array}
$
&
6
\\
\hline

$\displaystyle
\begin{array}{l}
  e^2\wedge e^5+e^3\wedge e^4
  \\
e^1\wedge e^5+e^2\wedge e^4
  \\
  e^1\wedge e^2
  \\
  e^2\wedge e^3
    \ \\
\text{(no.\ 52 in \cite{Gal})}

\end{array}
$
&
$\displaystyle
\begin{array}{c}
\\[0.01mm]
  g=\varphi_{1}(p^{32})^2 + 2\varphi_{2}p^{32}p^{21} + \varphi_{3}(p^{21})^2
  + 2\varphi_{4}(p^{51} + p^{42})p^{21} +
 2\varphi_{5}(p^{52}+p^{43})p^{32}
  \\[3mm]
\text{Singular variety  is a double plane and a double quadric,}
    \\[3mm]
u^2=0, ~~~ u^1u^3-(u^2)^2=0
\end{array}
$
&
6
\\
\hline

$\displaystyle
\begin{array}{l}
  e^1\wedge e^5+e^3\wedge e^4
  \\
e^2\wedge e^4
  \\
  e^1\wedge e^4+e^2\wedge e^3
  \\
  e^1\wedge e^3
    \ \\
\text{(no.\ 53 in \cite{Gal})}

\end{array}
$
&
$\displaystyle
\begin{array}{c}
\\[0.01mm]
  g=\varphi_{1}(p^{31})^2 + 2\varphi_{2}(p^{41} + 
p^{32})p^{31}
+ 2\varphi_{3}p^{42}p^{31} + \varphi_{3}(p^{41}+p^{32})^2
 \\
+ 2\varphi_{4}
(p^{41} +p^{32})p^{42}
+ \varphi_{5}(p^{42})^2 + 2\varphi_{6}(p^{51}+ p^{43})p^{31}
  \\[3mm]
\text{Singular\ variety is
a double\ plane and\ a\ double\ quadric,}
    \\[3mm]
u^1=0, ~~~ u^1u^4-u^2u^3=0
\end{array}
$
&
6
\\
\hline

$\displaystyle
\begin{array}{l}
  e^1\wedge e^3+e^2\wedge e^5
  \\
e^2\wedge e^4
  \\
  e^1\wedge e^5
  \\
  e^1\wedge e^4+e^2\wedge e^3
    \ \\
\text{(no.\ 54 in \cite{Gal})}

\end{array}
$
&
$\displaystyle
\begin{array}{c}
\\[0.01mm]
  g=2\varphi_{1}(p^{52}+p^{31})p^{42}+
  \varphi_{1}(p^{41}+p^{32})^2 +
 + 2\varphi_{2}(p^{41}+p^{32})p^{51} + \varphi_{2}(p^{52}+p^{31})^2
 \\
+  \varphi_{3}(p^{51})^2 + 2\varphi_{4}(p^{41}+p^{32})p^{42} + 2\varphi_{5}p^{51}p^{42} + 2\varphi_{5}(p^{52}+p^{31})(p^{41}+p^{32})
\\
+ \varphi_{6}(p^{42})^2 + 2\varphi_{7}(p^{52} + p^{31})p^{51}
  \\[3mm]
\text{Singular variety  is a double cubic,}
    \\[3mm]
u^1u^2u^3-(u^1)^2u^4-(u^2)^2=0
\end{array}
$
&
6
\\
\hline

$\displaystyle
\begin{array}{l}
  e^3\wedge e^5
  \\
e^1\wedge e^3+e^1\wedge e^4
  \\
  e^2\wedge e^4
  \\
  e^1\wedge e^2
    \ \\
\text{(no.\ 55 in \cite{Gal})}

\end{array}
$
&
$\displaystyle
\begin{array}{c}
\\[0.01mm]
  g=\varphi_{1}(p^{21})^2 + 2\varphi_{2}p^{42}p^{21} + \varphi_{3}(p^{42})^2
  + 2\varphi_{4}(p^{41} + p^{31})p^{21} +\\
 \varphi_{5}(p^{41}+p^{31})^2  + \varphi_{6}(p^{53})^2
  \\[3mm]
\text{Singular variety consists of 3 double planes,}
    \\[3mm]
u^1=0, ~~~ u^2=0, ~~~ u^3=0
\end{array}
$
&
6
\\
\hline

$\displaystyle
\begin{array}{l}
  e^2\wedge e^5+e^3\wedge e^4
  \\
e^1\wedge e^5+e^2\wedge e^4
  \\
  e^1\wedge e^3
  \\
  e^1\wedge e^2
    \ \\
\text{(no.\ 56 in \cite{Gal})}

\end{array}
$
&
$\displaystyle
\begin{array}{c}
\\[0.01mm]
  g=\varphi_{1}(p^{21})^2 + 2\varphi_{2}p^{31}p^{21} +
  \varphi_{3}(p^{31})^2
  + 2\varphi_{4}(p^{51}+ p^{42})p^{21}\\
   +
  2\varphi_{5}(p^{52}+p^{43})p^{21}
  + 2\varphi_{5}(p^{51}+p^{42})p^{31} 
  \\[3mm]
\text{Singular\ variety is
a  double plane and a double quadric,}
    \\[3mm]
u^1=0, ~~~ u^1u^3-(u^2)^2=0
\end{array}
$
&
7
\\
\hline

$\displaystyle
\begin{array}{l}
  e^1\wedge e^5
  \\
e^2\wedge e^4
  \\
  e^1\wedge e^4+e^2\wedge e^3
  \\
  e^1\wedge e^3
    \ \\
\text{(no.\ 58 in \cite{Gal})}

\end{array}
$
&
$\displaystyle
\begin{array}{c}
\\[0.01mm]
  g=\varphi_{1}(p^{31})^2 + 2\varphi_{2}(p^{41}+p^{32})p^{31}
+ 2\varphi_{3}p^{42}p^{31} + \varphi_{3}(p^{41}+p^{32})^2
\\
+ 2\varphi_{4}(p^{41}+p^{32})p^{42}
+ \varphi_{5}(p^{42})^2 + 2\varphi_{6}p^{51}p^{31} + \varphi_{7}(p^{51})^2
  \\[3mm]
\text{Singular variety  is
a double plane and a double quadric,}
    \\[3mm]
u^1=0, ~~~ u^1u^4-u^2u^3=0
\end{array}
$
&
7
\\
\hline

$\displaystyle
\begin{array}{l}
  e^2\wedge e^5+e^3\wedge e^4
  \\
e^1\wedge e^4
  \\
  e^1\wedge e^3
  \\
  e^2\wedge e^3
    \ \\
\text{(no.\ 59 in \cite{Gal})}

\end{array}
$
&
$\displaystyle
\begin{array}{c}
\\[0.01mm]
  g=\varphi_{1}(p^{32})^2 + 2\varphi_{2}p^{32}p^{31} + \varphi_{3}(p^{31})^2
+ 2\varphi_{4}p^{41}p^{31}\\
 + \varphi_{5}(p^{41})^2 + 2
\varphi_{6}(p^{52}+p^{43})p^{32}
  \\[3mm]
\text{Singular variety consists of 3 double planes,}
    \\[3mm]
u^1=0, ~~~ u^2=0, ~~~ u^3=0
\end{array}
$
&
7
\\
\hline

$\displaystyle
\begin{array}{l}
  e^2\wedge e^5+e^3\wedge e^4
  \\
e^1\wedge e^3+e^2\wedge e^3
  \\
  e^1\wedge e^4
  \\
  e^1\wedge e^2
    \ \\
\text{(no.\ 60 in \cite{Gal})}

\end{array}
$
&
$\displaystyle
\begin{array}{c}
\\[0.01mm]
  g=\varphi_{1}(p^{21})^2 + 2\varphi_{2}p^{41}p^{21} + \varphi_{3}
(p^{41})^2
+ 2\varphi_{4}(p^{32}+ p^{31})p^{21} \\
-2\varphi_{5}(p^{52}+p^{43})p^{21} + 2
\varphi_{5}(p^{32} + p^{31})p^{41}
+ \varphi_{6}(p^{32}+p^{31})^2
  \\[3mm]
\text{Singular\ variety\  consists\ of\ 3\  double\ planes,}
    \\[3mm]
u^1=0, ~~~ u^2=0, ~~~ u^1+u^2=0
\end{array}
$
&
8
\\
\hline

$\displaystyle
\begin{array}{l}
  e^4\wedge e^5
  \\
e^1\wedge e^4+e^2\wedge e^3
  \\
  e^1\wedge e^2
  \\
  e^1\wedge e^3
    \ \\
\text{(no.\ 63 in \cite{Gal})}

\end{array}
$
&
$\displaystyle
\begin{array}{c}
\\[0.01mm]
  g=\varphi_{1}(p^{31})^2 + 2\varphi_{2}p^{31}p^{21} + \varphi_{3}(p^{21})^2
+ 2\varphi_{4}(p^{41} + p^{32})p^{31}\\
 + 2\varphi_{5}(p^{41}+ p^{32})p^{21} + \varphi_{6}(p^{54})^2
  \\[3mm]
\text{Singular variety consists of 2 planes,}
    \\[3mm]
u^1=0~({\rm quadruple}), ~~~ u^4=0~ ({\rm double})
\end{array}
$
&
8
\\
\hline

$\displaystyle
\begin{array}{l}
  e^2\wedge e^4
  \\
e^1\wedge e^5
  \\
  e^1\wedge e^3
  \\
  e^2\wedge e^3
    \ \\
\text{(no.\ 64 in \cite{Gal})}

\end{array}
$
&
$\displaystyle
\begin{array}{c}
\\[0.01mm]
  g=\varphi_{1}(p^{32})^2 + 2\varphi_{2}p^{32}p^{31} + \varphi_{3}
(p^{31})^2
+ 2\varphi_{4}p^{51}p^{31} + \varphi_{5}(p^{51})^2 + 2
\varphi_{6}p^{42}p^{32}
\\
+ \varphi_{7}(p^{42})^2
  \\[3mm]
\text{Singular\ variety\  consists\ of\ 3\  double\ planes,}
    \\[3mm]
u^1=0, ~~~ u^2=0, ~~~ u^3=0
\end{array}
$
&
8
\\
\hline

$\displaystyle
\begin{array}{l}
  e^2\wedge e^5+e^3\wedge e^4
  \\
e^1\wedge e^4+e^2\wedge e^3
  \\
  e^1\wedge e^3
  \\
  e^1\wedge e^2
    \ \\
\text{(no.\ 65 in \cite{Gal})}

\end{array}
$
&
$\displaystyle
\begin{array}{c}
\\[0.01mm]
  g=\varphi_{1}(p^{21})^2 + 2\varphi_{2}p^{31}p^{21} + \varphi_{3}(p^{31})^2
+ 2\varphi_{4}(p^{41} + p^{32})p^{21} \\
+ 2\varphi_{5}(p^{41}
+ p^{32})p^{31} - 2
\varphi_{6}(p^{52} +p^{43})p^{21}
+ \varphi_{6}(p^{41}+p^{32})^2
  \\[3mm]
\text{Singular variety consists of 2 planes,}
    \\[3mm]
u^1=0~({\rm quadruple}), ~~~ u^2=0~ ({\rm double})
\end{array}
$
&
9
\\
\hline

$\displaystyle
\begin{array}{l}
  e^1\wedge e^3+e^2\wedge e^3
  \\
e^2\wedge e^4
  \\
  e^1\wedge e^5
  \\
  e^1\wedge e^2
    \ \\
\text{(no.\ 66 in \cite{Gal})}

\end{array}
$
&
$\displaystyle
\begin{array}{c}
\\[0.01mm]
  g=\varphi_{1}(p^{21})^2 + 2\varphi_{2}p^{51}p^{21} + \varphi_{3}(p^{51})^2
+ 2\varphi_{4}p^{42}p^{21} + \varphi_{5}(p^{42})^2 \\
+ 2
\varphi_{6}(p^{32}+ p^{31})p^{21} + 
\varphi_{7}(p^{32}+p^{31})^2
  \\[3mm]
\text{Singular variety consists of 3 double planes,}
    \\[3mm]
u^1=0, ~~~ u^2=0, ~~~ u^1+u^2=0
\end{array}
$
&
9
\\
\hline

$\displaystyle
\begin{array}{l}
  e^1\wedge e^5+e^2\wedge e^4
  \\
e^2\wedge e^3
  \\
  e^1\wedge e^4
  \\
  e^1\wedge e^3
    \ \\
\text{(no.\ 67 in \cite{Gal})}

\end{array}
$
&
$\displaystyle
\begin{array}{c}
\\[0.01mm]
  g=\varphi_{1}(p^{31})^2 + 2\varphi_{2}p^{41}p^{31} + \varphi_{3}(p^{41})^2
+ 2\varphi_{4}p^{32}p^{31} + 2\varphi_{5}(p^{51} + p^{42})p^{31}
\\
+ 2\varphi_{5}p^{41}p^{32} + 
\varphi_{6}(p^{32})^2 + 2\varphi_{7}(p^{51}
+ p^{42})p^{41}
  \\[3mm]
\text{Singular variety consists of 2 planes,}
    \\[3mm]
u^1=0~({\rm quadruple}), ~~~ u^3=0~ ({\rm double})
\end{array}
$
&
9
\\
\hline

$\displaystyle
\begin{array}{l}
  e^1\wedge e^5+e^2\wedge e^4
  \\
e^2\wedge e^3
  \\
  e^1\wedge e^4
  \\
  e^1\wedge e^2
    \ \\
\text{(no.\ 71 in \cite{Gal})}

\end{array}
$
&
$\displaystyle
\begin{array}{c}
\\[0.01mm]
  g=\varphi_{1}(p^{21})^2 + 2\varphi_{2}p^{41}p^{21} + \varphi_{3}
(p^{41})^2
+ 2\varphi_{4}p^{32}p^{21} + \varphi_{5}(p^{32})^2 \\
+ 2\varphi_{6}(p^{51}
+ p^{42})p^{21} + 2
\varphi_{7}(p^{51} + p^{42})p^{41}
  \\[3mm]
\text{Singular variety consists of 2 planes,}
    \\[3mm]
u^1=0~({\rm quadruple}), ~~~ u^2=0~ ({\rm double})
\end{array}
$
&
10
\\
\hline

$\displaystyle
\begin{array}{l}
  e^1\wedge e^5+e^2\wedge e^4
  \\
e^1\wedge e^4+e^2\wedge e^3
  \\
  e^1\wedge e^3
  \\
  e^1\wedge e^2
    \ \\
\text{(no.\ 73 in \cite{Gal})}

\end{array}
$
&
$\displaystyle
\begin{array}{c}
\\[0.01mm]
  g=\varphi_{1}(p^{21})^2 + 2\varphi_{2}p^{31}p^{21} + \varphi_{3}(p^{31})^2
+ 2\varphi_{4}(p^{41}+ p^{32})p^{21} + 2\varphi_{5}(p^{41}+p^{32})p^{31} \\
+ 2\varphi_{6}(p^{51} + p^{42})p^{31}
+ \varphi_{6}(p^{41}+p^{32})^2
+ 2\varphi_{7}(p^{51} + p^{42})p^{21}
  \\[3mm]
\text{Singular variety is a single six-tuple plane,}
    \\[3mm]
u^1=0
\end{array}
$
&
11
\\
\hline

$\displaystyle
\begin{array}{l}
  e^2\wedge e^3
  \\
e^1\wedge e^4
  \\
  e^1\wedge e^5
  \\
  e^1\wedge e^2
    \ \\
\text{(no.\ 79 in \cite{Gal})}

\end{array}
$
&
$\displaystyle
\begin{array}{c}
\\[0.01mm]
  g=\varphi_{1}(p^{21})^2 + 2\varphi_{2}p^{51}p^{21} + \varphi_{3}
(p^{51})^2
+ 2\varphi_{4}p^{41}p^{21} + 2\varphi_{5}p^{51}p^{41} + \varphi_{6}(p^{41})^2
\\
+ 2\varphi_{7}p^{32}p^{21} + \varphi_{8}(p^{32})^2
  \\[3mm]
\text{Singular variety consists of 2 planes,}
    \\[3mm]
u^1=0~({\rm quadruple}), ~~~ u^2=0~ ({\rm double})
\end{array}
$
&
12
\\
\hline

$\displaystyle
\begin{array}{l}
  e^1\wedge e^5+e^2\wedge e^3
  \\
e^1\wedge e^4
  \\
  e^1\wedge e^2
  \\
  e^1\wedge e^3
    \ \\
\text{(no.\ 86 in \cite{Gal})}

\end{array}
$
&
$\displaystyle
\begin{array}{c}
\\[0.01mm]
  g=\varphi_{1}(p^{31})^2 + 2\varphi_{2}p^{31}p^{21} + \varphi_{3}(p^{21})^2
+ 2\varphi_{4}p^{41}p^{31} + 2\varphi_{5}p^{41}p^{21} + \varphi_{6}(p^{41})^2
\\
+ 2\varphi_{7}(p^{51} + p^{32})p^{31} + 2\varphi_{8}(p^{51}
+ p^{32})p^{21}
  \\[3mm]
\text{Singular variety is a single six-tuple plane,}
    \\[3mm]
u^1=0
\end{array}
$
&
14
\\
\hline

$\displaystyle
\begin{array}{l}
  e^1\wedge e^2
  \\
e^1\wedge e^3
  \\
  e^1\wedge e^4
  \\
  e^1\wedge e^5
    \ \\
\text{(no.\ 97 in \cite{Gal})}

\end{array}
$
&
$\displaystyle
\begin{array}{c}
\\[0.01mm]
  g=\varphi_{1}(p^{51})^2 + 2\varphi_{2}p^{51}p^{41} + \varphi_{3}(p^{41})^2
+ 2\varphi_{4}p^{51}p^{31} + 2\varphi_{5}p^{41}p^{31} + \varphi_{6}(p^{31})^2
\\
+ 2\varphi_{7}p^{51}p^{21} + 2
\varphi_{8}p^{41}p^{21} + 2\varphi_{9}p^{31}p^{21}
+ \varphi_{10}(p^{21})^2
  \\[3mm]
\text{Singular variety is a single six-tuple plane,}
    \\[3mm]
u^1=0
\end{array}
$
&
20
\\
\hline
\end{longtable}
\end{center}

\noindent {\bf Remark 1.} Most of the cases in Table 1 can be uniquely
distinguished by dimensions of stabilisers and the geometry of singular
varieties.

\medskip
\noindent {\bf Remark 2.} Singular varieties appearing in Table 1 are nothing
but normal forms of determinantal cubics defined as
$$
\rk(A^1u, A^2u, A^3u, A^4u)<4,
$$
where $A^i$ are $5\times 5$ skew-symmetric matrices, and $u$ is a 5-component
column vector (note that all $4\times 4$ minors of this $5\times 4$ matrix have
one and the same cubic factor).

\medskip
\noindent {\bf Remark 3.} Monge metrics from Table 1 depend on auxiliary
parameters which, in some cases, can be removed by using the stabiliser (the
action of the stabiliser may have several non-equivalent orbits). Thus, one can
show that every Monge metric with 14-dimensional stabiliser is projectively
equivalent to one of the three canonical forms,
$$
\begin{array}{c}
(du^3+u^2du^1-u^1du^2)du^2+du^1du^4,\\
(du^3+u^2du^1-u^1du^2)du^2+(du^1)^2+(du^4)^2,\\
(du^3+u^2du^1-u^1du^2)du^2+(du^1)^2+du^2du^4;
\end{array}
$$
here the first canonical form corresponds to Example 6 of Sect
\ref{sec:examples}.

\subsection{Remarks on the multi-component case}
\label{sec:five-component-case}

For $n=5$ formula (\ref{A}) involves a 5-dimensional subspace $A=\langle A^1,
\dots, A^5 \rangle$ in the space of $6\times 6$ skew-symmetric forms,
equivalently, a point in $\Gr_5(\Lambda^2V^6)$.  Condition (\ref{King}) imposes
15 constraints for the 15 entries of $\phi$ (note that $\dim S^2A=\dim
\Lambda^4V^6=15$). For a {\it generic} subspace, these conditions are linearly
independent, and imply $\phi=0$. Thus, to get a nontrivial Monge metric one has
to select a subspace $A$ such that the forms $A^{\beta} \wedge A^{\gamma}$ are
linearly dependent. This defines a hypersurface $M$ in $\Gr_5(\Lambda^2V^6)$,
which is of degree 6 in the Pl\"ucker coordinates. Given a smooth generic point
of $M$, there exists a unique non-degenerate $\phi$ satisfying these
conditions. Thus, 5-component Hamiltonian operators are parametrised by points
of an algebraic hypersurface (of degree 6) in
$\Gr_5(\Lambda^2V^6)$. Unfortunately, it is highly unlikely that one can obtain
an effective classification of orbits of the associated ${ SL}(6)$-action, as
well as to classify Monge metrics corresponding to singular points of $M$.

\medskip
\noindent{\bf Example.} Consider the subspace spanned by the following
bivectors,
\begin{align*}
  & 2e^1\wedge e^2 + e^3\wedge e^4 + e^5\wedge e^6,
  \\
  &e^1\wedge e^3 + e^1\wedge e^4 + e^4\wedge e^6 + \alpha e^2\wedge e^4,
  \\
  &e^2\wedge e^6 - e^3\wedge e^5,
  \\
  &e^1\wedge e^6 - e^2\wedge e^3 + e^4\wedge e^5,
  \\
  &e^1\wedge e^5 + 2e^2\wedge e^5 +e^3\wedge e^4.
\end{align*}
One can show that the associated system (\ref{King}) has rank $15$ for
$\alpha\neq 0$, and therefore implies $\phi=0$; for $\alpha=0$ it has rank
$14$, and the corresponding (unique) nonzero solution $\phi$ is non-degenerate.

\medskip

For general $n$ the situation is similar: given a point in
$\Gr_n(\Lambda^2V^{n+1})$, condition (\ref{King}) imposes $C_{n+1}^4$
constraints for $C_n^2$ components of $\phi$. Thus, to get a nontrivial
solution one has to require that the number of independent constraints is less
than $C_n^2$. This defines an algebraic subvariety in $\Gr_n(\Lambda^2V^{n+1})$
whose geometry/singularities are yet to be investigated.

\section{Concluding remarks}

We have obtained a classification of third-order Hamiltonian operators of
differential-geometric type with the number of components $n\leq 4$. For $n=1,
2$ any such operator can be transformed to constant coefficients, for $n=3$ we
have 6 non-equivalent canonical forms, for $n=4$ there are 32 multi-parameter
families.

\begin{itemize}

\item Our approach is based on the classification of $SL(n+1)$-orbits in
  $\Gr_n(\Lambda^2V)$, which is only available for $n\le 4$. Apparently, for
  $n=5$ the problem becomes `wild', and no reasonable classification is
  possible.

\item Examples suggest that every homogeneous third-order Hamiltonian operator
  arises as a Hamiltonian structure of some {\it local} conservative system of
  hydrodynamic type of the form $u^i_t=(V^i(u))_x$, with non-local Hamiltonian
  (here $u^i$ are the flat coordinates). In the `generic' case, any system of
  this kind is linearly degenerate and non-diagonalisable.  Furthermore, in the
  generic case the fluxes $V^i$ are rational functions of the form
  $V^i=\frac{S^i}{S}$ where $S$ is the polynomial defining the singular variety
  of the corresponding Monge metric, and $S^i$ are polynomials of degree one
  higher.  It would be interesting to clarify the geometric meaning of such
  systems. It would also be interesting to classify higher-order conservative
  systems possessing homogeneous third-order Hamiltonian structures (the
  compatibility conditions between right-hand sides of such systems with the
  corresponding Hamiltonian operators are currently under investigation).

\end{itemize}

\section*{Acknowledgements}

We thank A. Bolsinov, R. Chiriv\`i, B. Doubrov, A. King, A.C. Norman,
A. Prendergast-Smith and D. Timashev for clarifying discussions. We acknowledge
financial support from GNFM of the Istituto Nazionale di Alta Matematica, the
Istituto Nazionale di Fisica Nucleare, and the Dipartimento di Matematica e
Fisica ``E. De Giorgi'' of the Universit\`a del Salento.  MVP's work was also
partially supported by the grant of Presidium of RAS \textquotedblleft
Fundamental Problems of Nonlinear Dynamics\textquotedblright\ and by the RFBR
grant 11-01-00197.

\end{document}